\documentclass[10pt, twocolumn, a4paper]{article}

\usepackage[utf8x]{inputenc} 
\usepackage{amsmath}
\usepackage{mathrsfs}
\usepackage{amsfonts}
\usepackage{amssymb}
\usepackage[dvips]{graphicx}
\usepackage{titletoc}
\usepackage{lmodern}
\usepackage[T1]{fontenc}
\usepackage{textcomp}
\usepackage{tikz}
\usepackage{color}
\usepackage{parskip}
\usepackage{etoolbox}
\usepackage{vmargin}

\setpapersize{A4}
\setmargins{1.4cm}       
{1cm}                        
{18.2cm}                      
{24cm}                    
{10pt}                           
{0.8cm}                           
{0pt}                             
{1.7cm}                           

\newtheorem{proposition}{Proposition}
\newenvironment{proof}{\par{\em{Proof: }}}{\hfill$\square$}
\newtheorem{remark}{Remark}

\newcommand{\sym}{ \mbox{\tt sym} }

\AtBeginEnvironment{array}{\setlength{\arraycolsep}{2pt}}
\AtBeginEnvironment{eqnarray}{\setlength{\arraycolsep}{2pt}}
\DeclareMathOperator{\sech}{sech}

\def\l2{{\mathcal L}_2}
\def\l2e{\call_{2e}}

\def\rea{\mathbb{R}}

\def\diag{\mbox{diag}}

\usepackage{color}

\begin{document}


\title{Saturated control without velocity measurements for planar robots with flexible joints} 

\author{ {\bf T. Wesselink, P. Borja, and J.M.A. Scherpen} \\
          Faculty of Science and Engineering,\\
        University of Groningen. Nijenborgh 4, 9747 AG Groningen, The Netherlands.\\
        e-mails:{\tt l.p.borja.rosales[j.m.a.scherpen]@rug.nl, t.c.wesselink@student.rug.nl}
        }
\date{}

\maketitle


\begin{abstract}
In this work, we propose a passivity-based controller that addresses the problem of set point regulation for planar robots with two links and flexible joints. 
Moreover, the controller is saturated and does not require velocity measurements. 
Additionally, we present experiments that corroborate the theoretical results of this note.
\end{abstract}


\section{Introduction}
\label{sec:int}
Energy-based models, e.g., the Euler-Lagrange (EL) and port-Hamiltonian (pH) frameworks, have been extensively used to represent mechanical systems, 
see for example \cite{BULLO, spong1987, olfati2001, ORTbook, VAN, ORTtac}. One of the main advantages of these modeling approaches is that they provide 
a systematic procedure to obtain mathematical models that capture the nonlinear phenomena and preserve conservation laws 
present in physical systems.

A well-known property of mechanical systems is that they are passive, loosely speaking this means that these systems are not able to generate energy by themselves. 
Particularly, in the EL and pH representations, this passivity property can be verified by considering as storage function the total energy of the system.
Accordingly, a natural way to control passive systems is to design controllers that give a desired shape to the energy of the closed-loop system. 
This process is known as energy-shaping and it is the main idea of several passivity-based control (PBC) approaches.

In this work, we focus on the pH representation of planar robots with flexible joints, where the objective is to address the problem of set-point regulation for this class of systems. Additionally, we are interested in controllers that can overcome two common issues that arise during practical implementation, namely, the lack of sensors to measure the velocities and limitations in the actuators, particularly, the necessity of saturated signals to ensure the safety of the equipment. Following these ideas, 
the main contribution of this work is the design of a controller that solves the set-point regulation problem for planar robots with flexible joints. The aforementioned controller 
has the following appealing properties:
\begin{itemize}
 \item The closed-loop system preserves the pH structure and, consequently, the passivity property.
 \item The control design does not require the solution of partial differential equations (PDEs).
 \item The control signals are constrained to a desired interval. Thus, for implementation purposes, it is not necessary to include additional saturation blocks to prevent damage to the motors.
 \item The control design only requires position measurements. Therefore, the controller can be implemented without the necessity of filters or observers that estimate the velocities.
\end{itemize}

The outline of this paper is as follows. First, we provide the model of the system and the problem formulation in Section \ref{sec:prem}. 
Section \ref{sec:cd} is devoted to the control design, where we present two saturated controllers and the experimental results derived from their implementation. Finally, we give closure to this note with some concluding
remarks and future work in Section \ref{sec:cfw}.\\[.2cm]
{\bf Notation:} We denote the $n \times n$ identity matrix as $I_n$, and the $n \times s$ matrix of zeros as $\mathbf{0}_{n \times s}$. Consider the vector $x \in \rea^n$, the square matrix $A \in \rea^{n \times n}$, the function $f:\rea^n \to \rea$, and the mapping $F:\rea^n \to \rea^m$. Then: 
we denote the $i-th$ element of $x$ as $x_{i}$. The symmetric part of $A$ is given by $\sym\{A\}:=\frac{1}{2}(A+A^{\top})$. When $A=A^{\top}$,
$A$ is said to be positive definite, $A>0$, or positive semi-definite, $A\geq 0$, if and only if $x^{\top}Ax>0$, $x^{\top}Ax\geq0$ for all $x\neq \mathbf{0}_{n}$, respectively. If $A>0$, we denote the Euclidean weighted-norm as $\lVert x \rVert^2_A:=x^\top A x$. We define the differential operator $\nabla_{x} f:=\left(\frac{\displaystyle \partial f }{\displaystyle \partial x}\right)^\top$ and $\nabla^2_x f:=\frac{\displaystyle \partial^2 f }{\displaystyle \partial x^2}$. 
For $F$, we define the $ij$-th element of its $n \times m$ Jacobian matrix as $(\nabla_x F)_{ij}:=\frac{\displaystyle \partial F_j}{\displaystyle \partial x_i}$.
When clear from the context the subindex in $\nabla$ is omitted. For any $F$ and the distinguished element $x_{*}  \in \rea^n$, we define the constant matrix $F_* :=F(x_* )$. 
All mappings are supposed smooth enough.
\section{Model and problem setting}
\label{sec:prem}

The system to be controlled --depicted in Fig. \ref{fig:planar}-- consists of two links, each one attached to a motor shaft through a spring. As is stated above, we adopt a pH model to characterize the behavior of the system.
Therefore, we consider as state vector the positions $q\in \rea^{4}$ and the momenta $p\in \rea^{4}$ related to the elements of the robot arm, namely,

\begin{equation}
\begin{array}{rl}
  q=\begin{bmatrix}
    q_{l} \\ q_{m}
   \end{bmatrix}, & p=\begin{bmatrix}
                       p_{l} \\ p_{m}
                      \end{bmatrix}
\end{array}
\end{equation} 
where the vectors $q_{l}\in\rea^{2}$, $q_{m}\in\rea^{2}$ denote the angular position of the links and the motors, respectively; while, $p_{l}\in\rea^{2}$ represent the momenta of the 
links, and the momenta of the motors are given by $p_{m}\in\rea^{2}$. Hence, the system dynamics is expressed as\footnote{For an alternative pH representation of this system, we refer the reader to \cite{jardon}.} 
\begin{equation}
\begin{array}{rcl}
  \begin{bmatrix}
  \dot{q} \\[.1cm] \dot{p}
 \end{bmatrix}
&=&\begin{bmatrix}
  \mathbf{0}_{4\times 4} & I_{4} \\[.1cm] -I_{4} & -R_{2} 
 \end{bmatrix}\begin{bmatrix}
               \nabla_{q} H(q,p) \\[.1cm] \nabla_{p} H(q,p)
              \end{bmatrix}
+\begin{bmatrix}
                \mathbf{0}_{4} \\[.1cm] B
               \end{bmatrix}u
\\[.5cm]
 H(q,p)&=&\frac{1}{2}p^{\top}M^{-1}(q_{l_{2}})p+\frac{1}{2}\lVert q_{l}-q_{m}\rVert_{K_{s}}^{2}
 \end{array}\label{sys}
\end{equation} 
with
\begin{equation}
 \begin{array}{rcl}
  M(q_{l_{2}})&=&\begin{bmatrix}
                M_{l}(q_{l_{2}}) & \mathbf{0}_{2\times 2} \\ \mathbf{0}_{2\times 2} & M_{m}  
               \end{bmatrix} \\ 
  M_{l}(q_{l_{2}})&=& \begin{bmatrix}
                       a_{1}+a_{2}+2b\cos(q_{l_{2}}) & a_{2}+b\cos(q_{l_{2}}) \\  a_{2}+b\cos(q_{l_{2}}) & a_{2}
                      \end{bmatrix} \\ 
  R_{2}&=&\begin{bmatrix}
           D_{l} & \mathbf{0}_{2\times 2} \\ \mathbf{0}_{2\times 2} & D_{m}
          \end{bmatrix}        \\ 
  M_{m}&=&\diag\{ \mathcal{I}_{m_{1}}, \mathcal{I}_{m_{2}} \} \\
  D_{l}&=& \diag\{ D_{l_{1}}, D_{l_{2}} \}\\
  D_{m}&=& \diag\{  D_{m_{1}}, D_{m_{2}}  \}\\
  K_{s}&=& \diag\{ k_{s_{1}}, k_{s_{2}}  \}\\
  B&=& \begin{bmatrix}
        \mathbf{0}_{2} \\ I_{2}
       \end{bmatrix}
 \end{array}
\end{equation} 
where $a_{1}, \ a_{2}$ and $b$ are constants related to the moment of inertia (MoI) of the links; and $u\in\rea^{2}$ is the input vector which corresponds to the torques of the motors.
All the parameters are positive and their physical meaning is explained in Table \ref{parameters} in the Appendix of this note.

\begin{figure}[h!]
 \centering
 \includegraphics[width=0.2\textwidth]{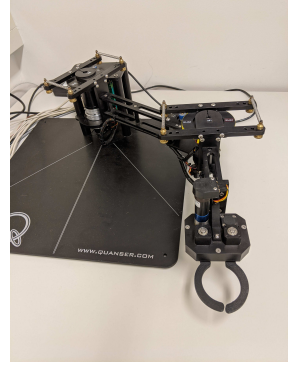}
 \caption{The planar robot with flexible joints by Quanser.}
 \label{fig:planar}
\end{figure}

\subsubsection*{Problem setting}
the objective of this work is to stabilize system \eqref{sys} to a constant point that belongs to the set
\begin{equation}
 \mathcal{E}:=\{ q\in\rea^{4}\mid q_{l}=q_{m}, p=\mathbf{0}_{4} \}.\label{eq}
\end{equation} 
Furthermore, we consider that only measurements of the positions $q$ are available; and the elements of the controller are constrained to given intervals, that is,
$u(t)\in \mathcal{U}$ for all $t\geq 0$, where
$\mathcal{U}:=[-u_{\tt max_{1}},u_{\tt max_{1}}]\times [-u_{\tt max_{2}},u_{\tt max_{2}}]$.

\section{Control design}
\label{sec:cd}
In this section, we present three controllers that stabilize the planar robot to a reference. The first controller is based on the PI-PBCs reported in \cite{BORCISORT}, which is used as a starting point towards the development of the saturated controllers. The second controller satisfies the requirements established in Section \ref{sec:prem}, nonetheless, the experiments exhibit a steady state error. Finally, the third controller is an extension of the previous one, where an integral-like action is added to eliminate the steady state error, this is corroborated in the experimental results.
\subsection{Preliminary PI controller}
In \cite{BORCISORT}, a constructive procedure to stabilize pH systems is proposed, an advantage of this approach over other PBC techniques is that the control law is obtained without the necessity of solving PDEs. 
Furthermore, the controllers derived from this approach can be interpreted as PI regulators, where the feedback signal is the passive output of the system. In Proposition \ref{prop:PI} we provide a modified PI that stabilizes system \eqref{sys} to a desired reference, and where the passive output is given by the velocities of the motors. 
Then, based on this PI controller, in subsequent sections we develop a control law that satisfies the requirements imposed in Section \ref{sec:prem}.
\begin{proposition}\label{prop:PI}
 Consider system \eqref{sys} in closed-loop with the controller\footnote{We recall that $\nabla_{p}H(q,p)=M^{-1}(q)p=\dot{q}$.}
 \begin{equation}
  u=-K_{P_{m}}\dot{q}_{m}-K_{I}(q_{m}-q_{*})-K_{P_{l}}\dot{q}_{l} \label{upi}
 \end{equation} 
 where $q_{*}\in \rea^{2}$ is the desired position of the links, and the matrices $K_{P_{m}},K_{P_{l}},K_{I}\in \rea^{2 \times 2}$ verify
 \begin{equation}
\begin{array}{rcl}
 K_{P_{m}}>0, \; \; \; K_{I}&>&0 \\[.1cm]
 D_{m}+K_{P_{m}}-\frac{1}{4}K_{P_{l}}^{\top}D_{l}^{-1}K_{P_{l}}&>&0.
\end{array}\label{pigains}
 \end{equation} 
 Then, the following statements hold true.
 \begin{itemize}
   \item [(i)] The closed-loop system admits a pH representation, that is
  \begin{equation}
   \begin{bmatrix}
    \dot{q} \\[.1cm] \dot{p} 
   \end{bmatrix}
=\begin{bmatrix}
  \mathbf{0}_{4\times 4}  & I_{4}  \\[.1cm]
   -I_{4} & J_{\tt PI_{2}}-R_{\tt PI_{2}} 
 \end{bmatrix}\begin{bmatrix}
               \nabla_{q}H_{\tt PI}(q,p) \\[.1cm] \nabla_{p}H_{\tt PI}(q,p) 
              \end{bmatrix} \label{clpi}
  \end{equation} 
  where
  \begin{eqnarray}
  \nonumber  R_{\tt PI_{2}}&:=&\begin{bmatrix}
                     D_{l} & \frac{1}{2}K_{P_{l}}^{\top} \\[.1cm] \frac{1}{2}K_{P_{l}} & D_{m}+K_{P_{m}} 
                    \end{bmatrix}, \\[.2cm] 
\nonumber J_{\tt PI_{2}}&:=&\frac{1}{2}\begin{bmatrix}
                 \mathbf{0}_{2\times 2} & K_{P_{l}}^{\top} \\[.1cm] -K_{P_{l}} & \mathbf{0}_{2\times 2}.
                \end{bmatrix}\\[.2cm]
H_{\tt PI}(q,p)&:=&H(q,p)+\frac{1}{2}\lVert q_{m}-q_{*}\rVert_{K_{I}}^{2}. \label{Hpi}
  \end{eqnarray}
  \item [(ii)] The point $$x_{*}:=(q_{l},q_{m},p_{l},p_{m})=(q_{*},q_{*},\mathbf{0}_{2},\mathbf{0}_{2})$$ is an asymptotically stable equilibrium of the closed-loop system with Lyapunov function $H_{\tt PI}(q,p)$.  
 \end{itemize}
\end{proposition}

\begin{proof}
 Note that, based on a Schur complement analysis,
 \eqref{pigains} implies that $R_{\tt PI_{2}}>0$. Moreover,
 \begin{equation}
  J_{\tt PI_{2}}-R_{\tt PI_{2}}= \begin{bmatrix}
                                  -D_{l} & \mathbf{0}_{2\times 2} \\
                                  -K_{P_{l}} & -D_{m}-K_{P_{m}}
                                 \end{bmatrix}.
 \end{equation} 
 Hence, replacing \eqref{upi} in \eqref{sys} we get\footnote{Due to space constraints, in the sequel, during the development of the proofs, the arguments of the functions are omitted when they are clear.}
 \begin{equation}
  \begin{array}{rcl}
    \dot{p}_{l}&=& -\nabla_{q_{l}}H-D_{l}\nabla_{p_{l}}H
    \\ \dot{p}_{m}&=&
-\nabla_{q_{m}}H-(D_{l}+K_{P_{m}})\nabla_{p_{m}}H-K_{P_{l}}\nabla_{p_{l}}H\\&&-K_{I}(q_{m}-q_{*}).
  \end{array}\label{dppi}
 \end{equation} 
 Therefore, since 
 \begin{equation*}
  \begin{array}{rl}
\nabla_{q}H_{\tt PI}=\nabla_{q}H+\begin{bmatrix}\mathbf{0}_{2} \\  K_{I}(q_{m}-q_{*})\end{bmatrix}, &  \nabla_{p}H_{\tt PI}=\nabla_{p}H,
  \end{array}
 \end{equation*} 
  we claim that (i) holds true. To proof (ii), note that
  \begin{equation}
   \dot{H}_{\tt PI}=\left[
(\nabla_{q} H_{\tt PI})^{\top} \ (\nabla_{p} H_{\tt PI})^{\top}                     
                    \right]
\begin{bmatrix}
 \dot{q} \\ \dot{p}
\end{bmatrix}=-\lVert \nabla_{p} H_{\tt PI} \rVert^{2}_{R_{\tt PI_{2}}},\label{dHpi}
  \end{equation} 
which implies that ${H}_{\tt PI}$ is non-increasing. Moreover, some straightforward computations show that
\begin{equation}
 \begin{array}{l}
  (\nabla H_{\tt PI})_{*} = \mathbf{0}_{8} \\[.2cm]
  (\nabla^{2} H_{\tt PI})_{*}=\begin{bmatrix}
                               K_{\tt PI} & \mathbf{0}_{4\times 4} \\
                               \mathbf{0}_{4\times 4} & M_{*}
                              \end{bmatrix}>0
 \end{array}\label{minPI}
\end{equation}
where
\begin{equation*}
 K_{\tt PI}:=\begin{bmatrix}
             K_{s} & -K_{s} \\ -K_{s} & K_{s}+K_{I}
            \end{bmatrix}.
\end{equation*} 
Therefore, ${H}_{\tt PI}$ has a minimum at $x_{*}$. Thus, the stability of the equilibrium is proved by applying the direct method of Lyapunov, see \cite{KHA}. Furthermore, 
\begin{equation}
 \begin{array}{l}
  \dot{H}_{\tt PI}=0 \Longleftrightarrow p=\mathbf{0}_{4} \implies \dot{p}=\mathbf{0}_{4}
  \\[.1cm]\implies \left. \begin{array}{l}K_{s}(q_{l}-q_{m})=\mathbf{0}_{2}  \\  K_{I}(q_{m}-q_{*})=\mathbf{0}_{2}                               \end{array} 
 \right\rbrace \implies q_{m}=q_{l}=q_{*}.
 \end{array}
\end{equation} 
Hence, the asymptotic stability property is proved by Barbashin Theorem, see \cite{KHA}.
  \end{proof}
\begin{remark}
 The natural damping of the system ensures that the inclusion of the term $-K_{P_{l}}\dot{q_{l}}$ is possible at the same time that the pH structure is preserved. On one hand, the aforementioned term injects
 damping to the links, which attenuate the oscillations. On the other hand, the structure preservation is desirable for analysis purposes, and potentially, physical interpretation of the closed-loop system.
\end{remark} 
\subsection{Saturated control without velocity measurements}
Although the controller defined in \eqref{upi} stabilizes the system at the desired point, it clearly requires information of the velocities. 
Moreover, it is not possible to ensure that the control signals remain in the range of operation of the motors. 
Accordingly, to overcome the aforementioned issues, we propose two modifications to the control law:
\begin{itemize}
 \item To replace the integral term $K_{I}(q_{m}-q_{*})$ with a saturated function.
 \item To inject damping without the necessity of measuring the velocities. 
\end{itemize}
A method to inject damping without velocity measurements, for mechanical systems, is proposed in \cite{dirksz}. The main idea of this methodology is to propose a virtual state that is linearly related to the positions, 
then, this new state is used to inject damping into the closed-loop system. Proposition \ref{prop:sat} establishes one of the main contributions of this note, where a combination of the damping injection approach 
reported in \cite{dirksz} with the PI of Proposition \ref{prop:PI} is proposed. 
\begin{proposition}\label{prop:sat}
Let the controller state vectors $x_{c_{l}},x_{c_{m}}\in \rea^{2}$. Define the functions\footnote{Due to space constraints, the arguments $(q_{l},x_{c_{l}}),(q_{m},x_{c_{m}})$ of functions $z_{l},z_{m}$, respectively, are omitted in the definitions of $\Phi_{l}(z_{l}(q_{l},x_{c_{l}}),\Phi_{m}(z_{m}(q_{m},x_{c_{m}})$.}
\begin{equation}
 \begin{array}{rl}
  z_{l}(q_{l},x_{c_{l}}):=&q_{l}-q_{*}+x_{c_{l}}, \\ z_{m}(q_{m},x_{c_{m}}):=&q_{m}-q_{*}+x_{c_{m}}, \\ 
  \Phi_{l}(z_{l}):=&\displaystyle\sum_{i=1}^{2}\frac{\alpha_{l_{i}}}{\beta_{l_{i}}}\ln(\cosh(\beta_{l_{i}}z_{l_{i}}))  \\
  \Phi_{m}(z_{m}):=&\displaystyle\sum_{i=1}^{2}\frac{\alpha_{m_{i}}}{\beta_{m_{i}}}\ln(\cosh(\beta_{m_{i}}z_{m_{i}}))
 \end{array}\label{zPhi}
\end{equation} 
where $\alpha_{l_{i}}, \alpha_{m_{i}}, \beta_{l_{i}}, \beta_{m_{i}}\in \rea_{>0}$.
Consider the dynamics
\begin{equation}
 \begin{array}{rcl}
 \dot{x}_{c_{l}} &=& -R_{c_{l}}\nabla_{x_{c_{l}}}\Phi_{l}(z_{l}(q_{l},x_{c_{l}})) \\[.1cm] \dot{x}_{c_{m}}&=&-R_{c_{m}}\left[  \nabla_{x_{c_{m}}}\Phi_{m}(z_{m}(q_{m},x_{c_{m}}))+K_{c}x_{c_{m}}\right],
 \end{array}\label{xc}
\end{equation} 
where $R_{c_{l}}, R_{c_{m}}, K_{c}\in \rea^{2\times 2}$ are positive definite constant matrices verifying
\begin{equation}
R_{c_{l}}-\frac{1}{4} \left(D_{l}^{-1}+D_{m}^{-1}\right)>0.  \label{schurzeta}
\end{equation} 
Consider the control law 
\begin{equation}
u=-\nabla_{z_{l}}\Phi_{l}(z_{l}(q_{l},x_{c_{l}}))-\nabla_{z_{m}}\Phi_{m}(z_{m}(q_{m},x_{c_{m}})). \label{uzeta}
\end{equation} 
Then:
\begin{itemize}
 \item[(i)] The elements of the input vector $u$ are saturated.
 \item[(ii)] Consider system \eqref{sys} in closed-loop with \eqref{uzeta}. Hence, the dynamics of the augmented state space $\zeta=[q^{\top},p^{\top},x_{c_{l}}^{\top},x_{c_{m}}^{\top}]^{\top}$ admit a pH representation.
 \item[(iii)] The point $\zeta_{*}=(x_{*},\mathbf{0}_{2},\mathbf{0}_{2})$ is an asymptotically stable equilibrium of the closed loop system with Lyapunov function
 \begin{equation}
 \begin{array}{rcl}
  H_{\zeta}(\zeta)&=&H(q,p)+\Phi_{l}(z_{l}(q_{l},x_{c_{l}}))\\&&+\Phi_{m}(z_{m}(q_{m},x_{c_{m}}))+\frac{1}{2}\lVert x_{c_{m}}\rVert^{2}_{K_{c}}. \label{Hzeta}
 \end{array}
 \end{equation} 
\end{itemize}
\end{proposition}
\begin{proof}
To proof (i), note that
\begin{equation}
 \begin{array}{rcl}
  \nabla_{z_{l}}\Phi_{l}&=&  \begin{bmatrix}
  \alpha_{l_{1}}\tanh(\beta_{l_{1}}z_{l_{1}}) \\ \alpha_{l_{2}}\tanh(\beta_{l_{2}}z_{l_{2}})
                                                    \end{bmatrix}\\[.35cm] \nabla_{z_{m}}\Phi_{m}&=& \begin{bmatrix}
  \alpha_{m_{1}}\tanh(\beta_{m_{1}}z_{m_{1}}) \\ \alpha_{m_{2}}\tanh(\beta_{m_{2}}z_{m_{2}})
                                                    \end{bmatrix}. 
 \end{array}
\end{equation} 
Therefore, the control law \eqref{uzeta} reduces to
\begin{equation}
 \begin{bmatrix}
  u_{1} \\ u_{2}
 \end{bmatrix}
=\begin{bmatrix}
  \alpha_{l_{1}}\tanh(\beta_{l_{1}}z_{l_{1}})+\alpha_{m_{1}}\tanh(\beta_{m_{1}}z_{m_{1}}) \\ \alpha_{l_{2}}\tanh(\beta_{l_{2}}z_{l_{2}})+\alpha_{m_{2}}\tanh(\beta_{m_{2}}z_{m_{2}})
\end{bmatrix}.
\end{equation} 
Thus,
\begin{equation}
  -( \alpha_{l_{i}}+\alpha_{m_{i}})\leq u_{i} \leq  \alpha_{l_{i}}+\alpha_{m_{i}}.
 \end{equation} 
To proof (ii) note that, since
\begin{equation}
 \nabla_{q_{l}}z_{l}=\nabla_{x_{c_{l}}}z_{l}=\nabla_{q_{m}}z_{m}=\nabla_{x_{c_{m}}}z_{m}=I_{2},
\end{equation} 
from the chain rule we get
\begin{equation}
 \begin{array}{l}
  \nabla_{q_{l}}\Phi_{l}=\nabla_{x_{c_{l}}}\Phi_{l}=\nabla_{x_{c_{l}}}H_{\zeta}=\nabla_{z_{l}}\Phi_{l}
, \\[0.35cm] \nabla_{q_{m}}\Phi_{m}=\nabla_{x_{c_{m}}}\Phi_{m}=\nabla_{z_{m}}\Phi_{m}.
 \end{array}
\end{equation}
Therefore, in closed-loop, the dynamics of the momenta vector take the form 
\begin{equation}
 \begin{array}{rcl}
  \dot{p}_{l}&=&-\nabla_{q_{l}}H-D_{l}\nabla_{p}H\\&=&-\nabla_{q_{l}}H_{\zeta}-D_{l}\nabla_{p}H_{\zeta}+\nabla_{x_{c_{l}}}H_{\zeta} \\
  \dot{p}_{m}&=&-\nabla_{q_{m}}H-D_{m}\nabla_{p_{m}}H+u \\
  &=&-\nabla_{q_{m}}H_{\zeta}-D_{m}\nabla_{p_{m}}H_{\zeta}-\nabla_{x_{c_{l}}}H_{\zeta}.
 \end{array}\label{dpzeta}
\end{equation} 
Moreover, \eqref{xc} can be rewritten as
\begin{equation}
\begin{array}{rcl}
 \dot{x}_{c_{l}} &=& -R_{c_{l}}\nabla_{x_{c_{l}}}H_{\zeta} \\[.1cm] \dot{x}_{c_{m}}&=&-R_{c_{m}}\nabla_{x_{c_{m}}}H_{\zeta}. \label{dxczeta}
 \end{array} 
\end{equation}  
Hence, from \eqref{dpzeta} and \eqref{dxczeta}, the closed-loop system takes the form
\begin{equation}
\dot{\zeta}=F_{\zeta}\nabla H_{\zeta}, \label{clzeta}
\end{equation} 
with,
\begin{equation}
 F_{\zeta}:=\begin{bmatrix}
             \mathbf{0}_{4\times 4} & I_{4} & \mathbf{0}_{4\times 2} & \mathbf{0}_{4\times 2} \\[.1cm] -I_{4} & -R_{2} & \Gamma_{\zeta} & \mathbf{0}_{4\times 2} \\[.1cm] \mathbf{0}_{2\times 4} & \mathbf{0}_{2\times 4} & -R_{c_{l}} & \mathbf{0}_{2\times 2} \\[.1cm] \mathbf{0}_{2\times 4} & \mathbf{0}_{2\times 4}  & \mathbf{0}_{2\times 2} & -R_{c_{m}}
            \end{bmatrix};\; \; \Gamma_{\zeta}:=\begin{bmatrix}
                  I_{2} \\ -I_{2}
                 \end{bmatrix}.\label{Fzeta}
\end{equation}  
Furthermore, from \eqref{schurzeta}, it follows that
\begin{equation}
 \sym\{ F_{\zeta}\}=\begin{bmatrix}
\mathbf{0}_{4\times 4} & \mathbf{0}_{4\times 4} & \mathbf{0}_{4\times 2} & \mathbf{0}_{4 \times 2} \\[.1cm] \mathbf{0}_{4\times 4} & -R_{2} & \frac{1}{2}\Gamma_{\zeta} & \mathbf{0}_{4\times 2}\\[.1cm] \mathbf{0}_{2\times 4} & \frac{1}{2}\Gamma_{\zeta}^{\top} & -R_{c_{l}} & \mathbf{0}_{2\times 2} \\[.1cm] \mathbf{0}_{2\times 4} & \mathbf{0}_{2\times 4}  & \mathbf{0}_{2\times 2}^{\top} & -R_{c_{m}}
                    \end{bmatrix}\leq0. \label{Rzeta}
\end{equation} 
To proof (iii) note that, from \eqref{clzeta} and \eqref{Rzeta}, we have
\begin{equation}
 \begin{array}{rcl}
  \dot{H}_{\zeta}=(\nabla H_{\zeta})^{\top}\dot{\zeta}=(\nabla H_{\zeta})^{\top}F_{\zeta}\nabla H_{\zeta}\leq 0,
 \end{array}
\end{equation} 
which implies that $H_{\zeta}(\zeta)$ is non-increasing. Moreover, 
\begin{equation}
 \begin{array}{rcl}
  z_{l_{*}}=\mathbf{0}_{2} & \implies & (\nabla_{q_{l}}\Phi_{l})_{*}=(\nabla_{x_{c_{l}}}\Phi_{l})_{*}=\mathbf{0}_{2}, \\
  z_{m_{*}}=\mathbf{0}_{2} & \implies & (\nabla_{q_{m}}\Phi_{m})_{*}=(\nabla_{x_{c_{m}}}\Phi_{m})_{*}=\mathbf{0}_{2}.
 \end{array}
\end{equation} 
Therefore,
\begin{equation}
 (\nabla H_{\zeta})_{*}=\begin{bmatrix}
                         (\nabla H)_{*} \\ \mathbf{0}_{2} \\ K_{c}x_{c_{m_{*}}}
                        \end{bmatrix}=\mathbf{0}_{12}. \label{gradzeta}
\end{equation} 
Furthermore, since
\begin{equation*}
 \begin{array}{l}
  \nabla_{q_{l}}^{2}\Phi_{l}=\nabla_{x_{c_{l}}}(\nabla_{q_{l}}\Phi_{l})=\nabla_{q_{l}}(\nabla_{x_{c_{l}}}\Phi_{l})=\nabla_{x_{c_{l}}}^{2}\Phi_{l}\\
 =\diag\{\beta_{l_{1}}\alpha_{l_{1}}\sech^{2}(\beta_{l_{1}}z_{l_{1}}), \beta_{l_{2}}\alpha_{l_{2}}\sech^{2}(\beta_{l_{2}}z_{l_{2}}) \} \\[.1cm]
  \nabla_{q_{m}}^{2}\Phi_{m}=\nabla_{x_{c_{m}}}(\nabla_{q_{m}}\Phi_{m})=\nabla_{q_{m}}(\nabla_{x_{c_{m}}}\Phi_{m})=\nabla_{x_{c_{m}}}^{2}\Phi_{m}\\
 =\diag\{\beta_{m_{1}}\alpha_{m_{1}}\sech^{2}(\beta_{m_{1}}z_{m_{1}}), \beta_{m_{2}}\alpha_{m_{2}}\sech^{2}(\beta_{m_{2}}z_{m_{2}}) \},
 \end{array}
\end{equation*}
some straightforward computations show that
\begin{equation}
 (\nabla^{2} H_{\zeta})_{*}=\begin{bmatrix}
                             K_{\tt S}+A  & \mathbf{0}_{4 \times 4} & A\\
                             \mathbf{0}_{4 \times 4} &  M_{*}^{-1} & \mathbf{0}_{4 \times 4}  \\ A & \mathbf{0}_{4 \times 4} & A+K_{\tt C}
                            \end{bmatrix}>0,\label{hesszeta}
\end{equation}
where
\begin{equation*}
\begin{array}{l}
  A:=\diag\{ \beta_{l_{1}}\alpha_{l_{1}}, \beta_{l_{2}}\alpha_{l_{2}}, \beta_{m_{1}}\alpha_{m_{1}}, \beta_{m_{1}}\alpha_{m_{2}} \},\\[.1cm]
  \begin{array}{rl}
   K_{\tt S}:=\begin{bmatrix}
             K_{s} & -K_{s} \\ -K_{s} & K_{s}
            \end{bmatrix}, & 
K_{\tt C}:=\begin{bmatrix}
              \mathbf{0}_{2\times 2} & \mathbf{0}_{2\times 2} \\ \mathbf{0}_{2\times 2} & K_{c}
             \end{bmatrix}.
  \end{array}
\end{array}
\end{equation*} 
Accordingly, from \eqref{gradzeta} and \eqref{hesszeta}, $\arg\min\{ H_{\zeta}(\zeta) \}=\zeta_{*}$. Thus, the stability property of $\zeta_{*}$ is proved by invoking Lyapunov theory. Moreover, to proof the asymptotic stability property, note that
\begin{equation}
\begin{array}{l}
  \dot{H}_{\zeta}=0 \Longleftrightarrow \left\lbrace\begin{array}{rcl}
   \nabla_{p}H_{\zeta}&=&\mathbf{0}_{4}\Longrightarrow p=\mathbf{0}_{4} \\ \nabla_{x_{c_{l}}}H_{\zeta}&=&\mathbf{0}_{2} \Longrightarrow z_{l}=\mathbf{0}_{2} \\ \nabla_{x_{c_{m}}}H_{\zeta}&=&\mathbf{0}_{2} \Longrightarrow K_{c}x_{c_{m}}=-\nabla_{q_{m}}\Phi_{m}                                   
                                        \end{array}\right.\\[.6cm]
 \Longrightarrow\left\lbrace \begin{array}{rcl}
 \nabla_{q_{l}}\Phi_{l}&=&\mathbf{0}_{2}\\
 \dot{p}_{l}&=&\mathbf{0}_{2}\\ \dot{p}_{m}&=&\mathbf{0}_{2}              
              \end{array} \right\rbrace\Longrightarrow 
 \left\lbrace \begin{array}{rcl}
              -K_{s}(q_{l}-q_{m})&=&\mathbf{0}_{2} \\ x_{c_{m}}&=&\mathbf{0}_{2} 
              \end{array}\right\rbrace \\[.6cm]             
               \Longrightarrow \left\lbrace \begin{array}{rcl}
                                             q_{m}&=&q_{l} \\ 
                                             z_{m}&=&\mathbf{0}_{2}
                                            \end{array} \right\rbrace                            
\Longrightarrow \left\lbrace \begin{array}{l}
                                         q_{m}=q_{l}=q_{*} \\ x_{c_{l}}=\mathbf{0}_{2}.
                                        \end{array}\right.
                                            \end{array}\label{aszeta}
\end{equation} 
Hence, the asymptotic stability property is proved by Barbashin Theorem.
\end{proof}
We corroborate the analytical results of Proposition \ref{prop:sat} through the implementation of the controller in a physical system. Below we report the obtained experimental results. 
\subsubsection*{\bf{Experimental results}}
the controller \eqref{uzeta} is implemented in the robot arm \textit{2 DOF serial flexible joint} by Quanser. We fix the following control parameters.  
\begin{equation}
 \begin{array}{rccccl}
  R_{c_{m}}&=&\diag\{ 25, 25\},&  \beta_{l_{1}}&=&2,\\[.1cm]  R_{c_{l}}&=&\diag\{ 10, 40\}, & \beta_{l_{2}}&=&1,  \\[.1cm] K_{c}&=&\diag\{ 5, 5\}, &  \beta_{m_{i}}&=&1;
 \end{array}\label{gains}
\end{equation}
and we study three cases for different values of the parameters $\alpha_{l_{i}}, \ \alpha_{m_{i}}$ which are shown in Table \ref{acase}. 
\begin{table}[!h]
\caption{Cases for different values of the parameters $\alpha$}
\label{acase}
\vspace{.2cm}
\begin{center}
\begin{tabular}{|c|c|c|}
\hline
Case & values for $\alpha_{l_{i}}$ & values for $\alpha_{m_{i}}$  \\\hline
\textbf{C1} & $0.8$ & $0.4$ \\\hline
\textbf{C2} & $0.2$ & $1$ \\\hline
\textbf{C3} & $0$ & $1.2$ \\\hline
\end{tabular}
\end{center}
\end{table}

The three cases under study illustrate the effect of the term $\nabla_{z_{l}}\Phi_{l}(z_{l}(q_{l},x_{c_{l}}))$ in the behavior of the closed-loop system. 
This term is interpreted as damping in the links of the planar robot, therefore, is expected that as the values $\alpha_{l_{i}}$ increase, the oscillations in the response decrease. 
Accordingly, the experiments for the three cases are carried out under the same initial conditions, $\zeta(0)=\mathbf{0}_{12}$, and the same reference, $q_{*}=(-1,1)$. Figures \ref{fig:az}, \ref{fig:u_azero} depict the
results of case \textbf{C1}, Figures \ref{fig:as}, \ref{fig:u_as} show the results of case \textbf{C2}; and Figures \ref{fig:al}, \ref{fig:u_al} correspond to the results of case \textbf{C3}. From the aforementioned plots,
we conclude the following:
\begin{itemize}
 \item The existence of steady state error can be observed for all the cases. This situation is probably a consequence of phenomena that are not taken into account in the model, e.g., nonlinear friction terms. Moreover, from the experiments,
 we notice the actuators are not able to provoke any displacement when $u_{i}\in[-0.12,0.12]$. Therefore, the 
 positions remain constant, even when the control signals are different from zero. This is particularly notorious in Figure \ref{fig:u_al}.
 \item There exists a trade off between the damping injected to the links and the magnitude of the steady state error. Furthermore,
 a similar relationship takes place between the magnitude of $\alpha_{l_{i}}$ and the oscillations, where, a greater magnitude of these values yields into an important attenuation in the oscillations.
 Both relations can be noticed in Figures \ref{fig:al},
 \ref{fig:as} and \ref{fig:az}.
 \end{itemize}
\begin{figure}
 \centering
 \includegraphics[width=.4\textwidth]{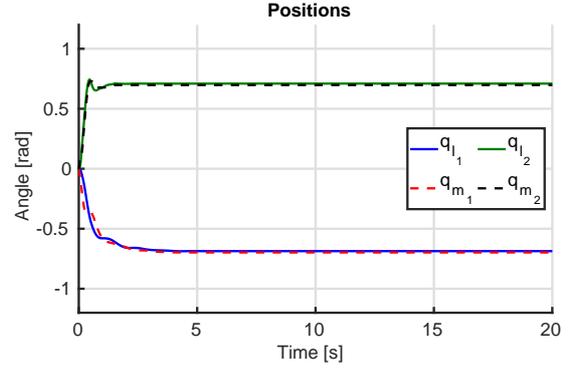}
 \caption{Behavior of the positions for \textbf{C1}.}
 \label{fig:al}
\end{figure}
\begin{figure}
 \centering
 \includegraphics[width=.4\textwidth]{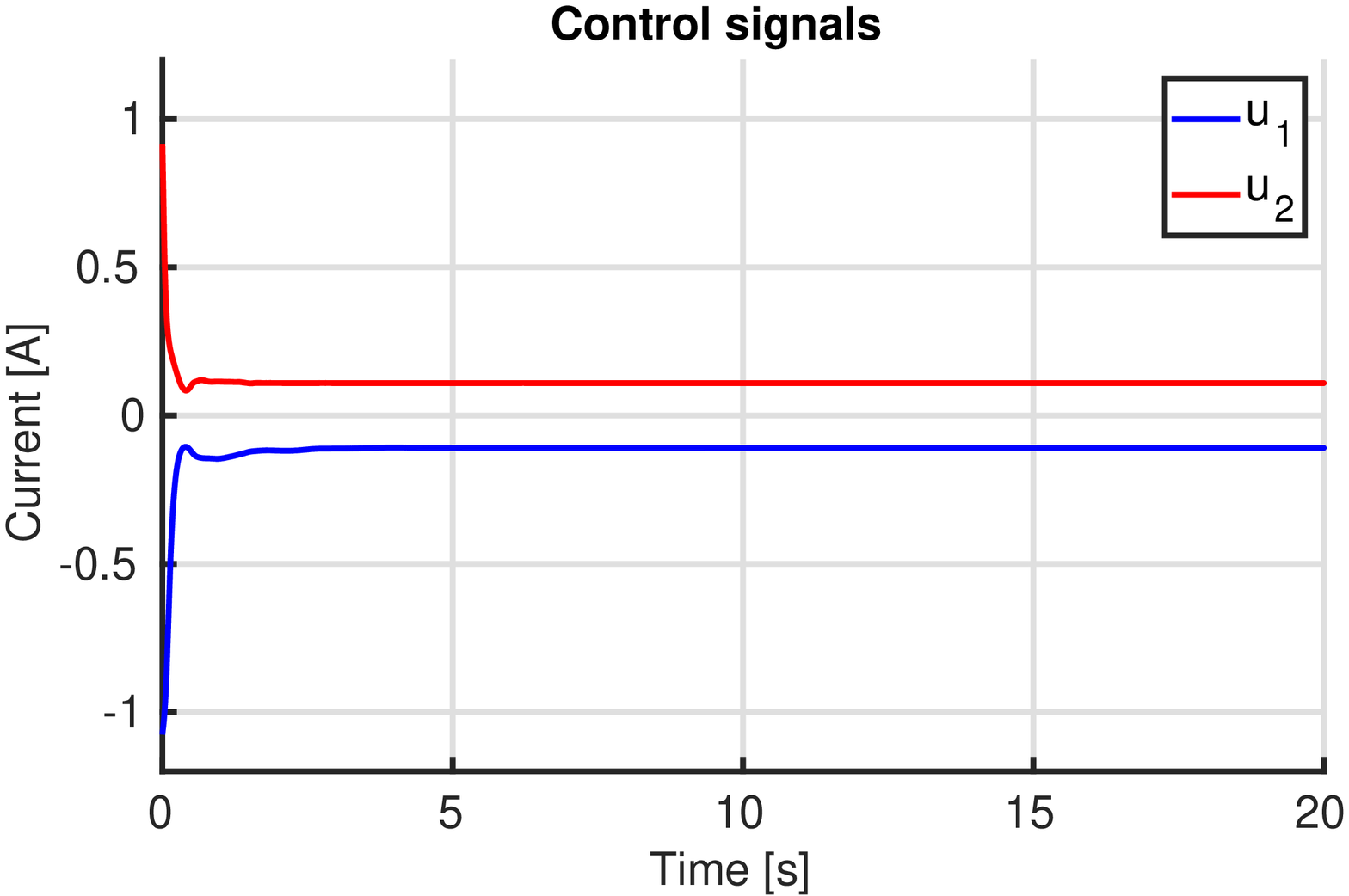}
 \caption{Control signals for \textbf{C1}.}
 \label{fig:u_al}
\end{figure}
\begin{figure}
 \centering
 \includegraphics[width=.4\textwidth]{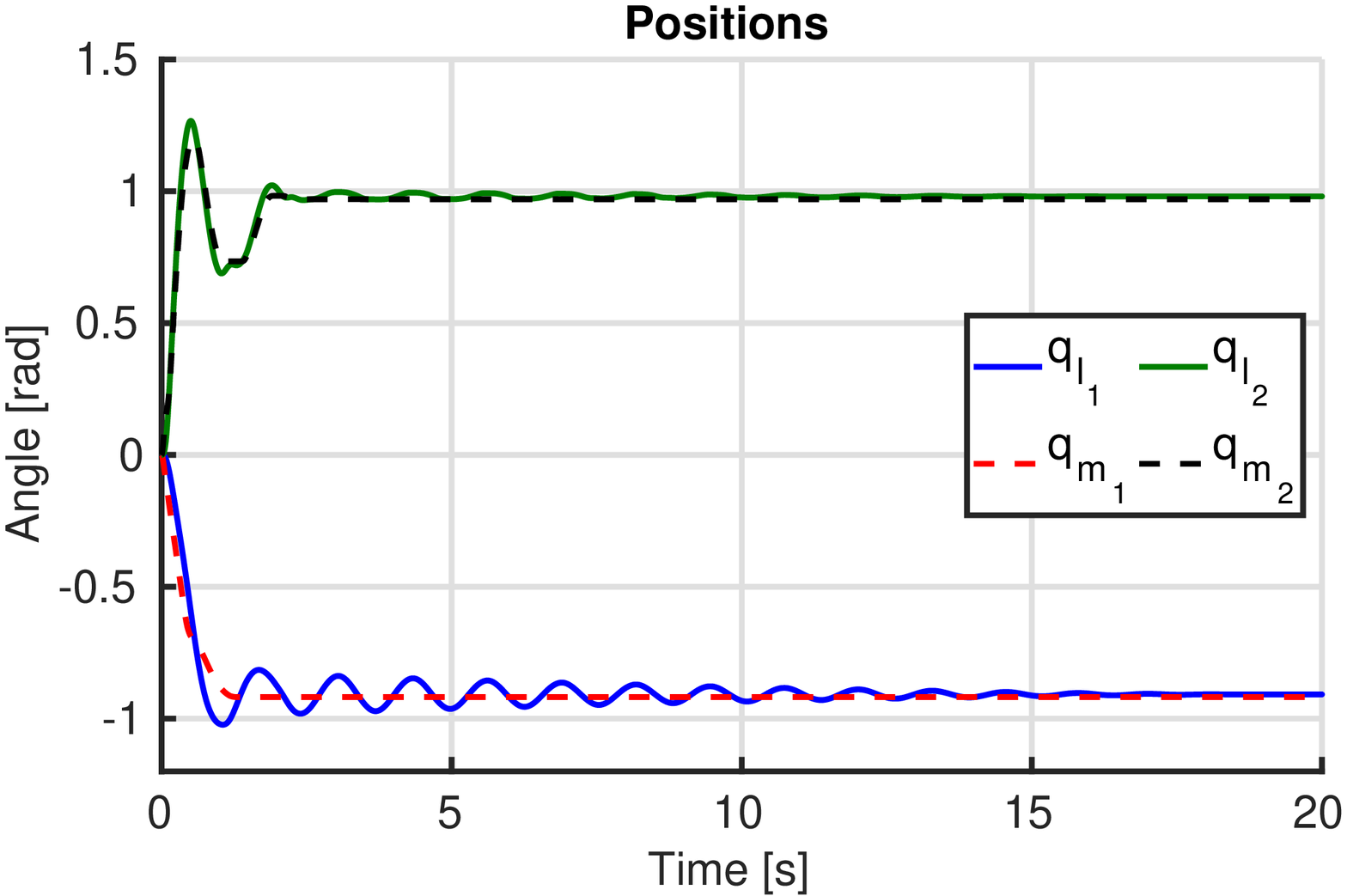}
 \caption{Behavior of the positions for \textbf{C2}.}
 \label{fig:as}
\end{figure}
\begin{figure}
 \centering
 \includegraphics[width=.4\textwidth]{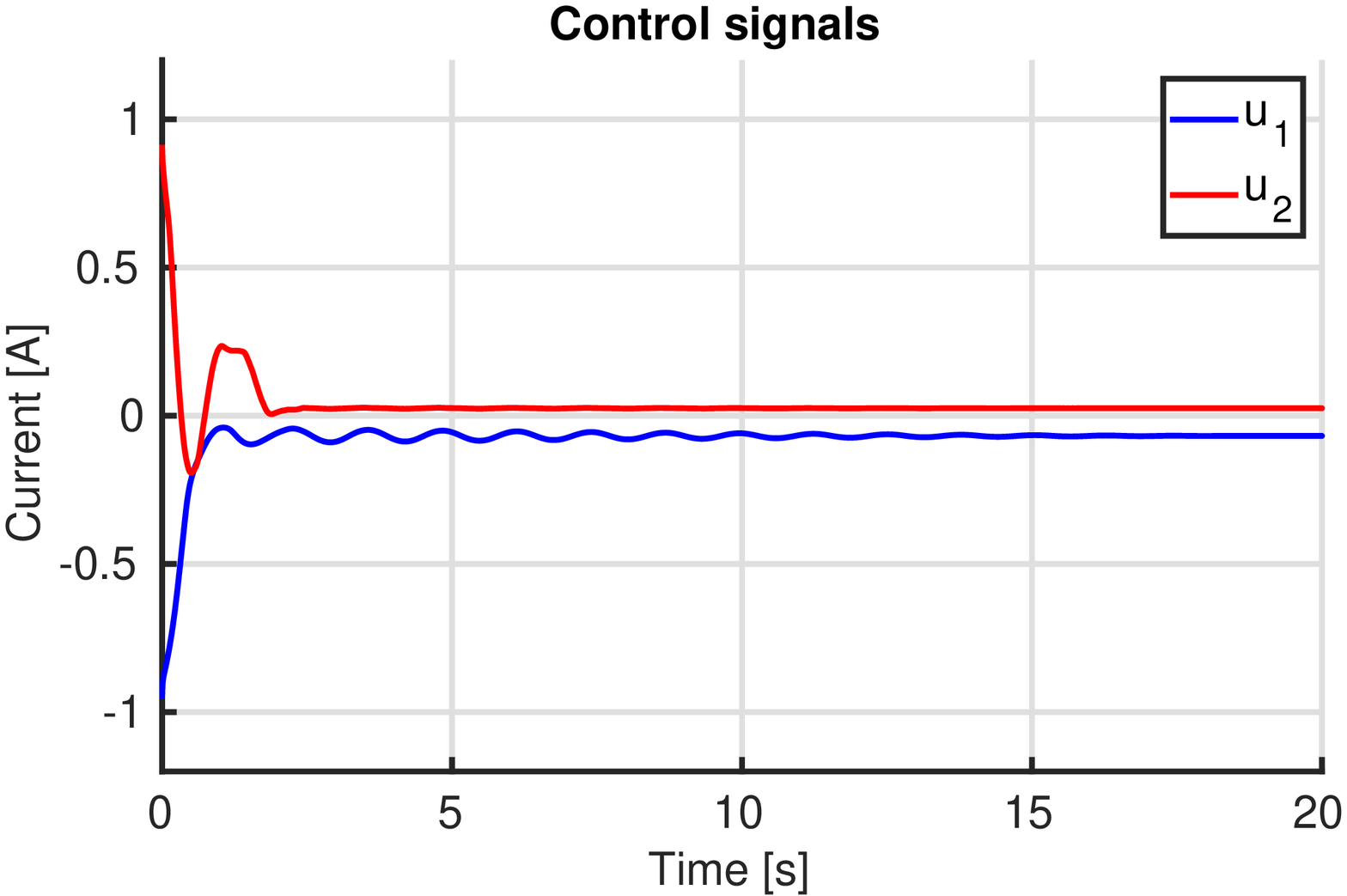}
 \caption{Control signals for \textbf{C2}.}
 \label{fig:u_as}
\end{figure}
\begin{figure}
 \centering
 \includegraphics[width=.4\textwidth]{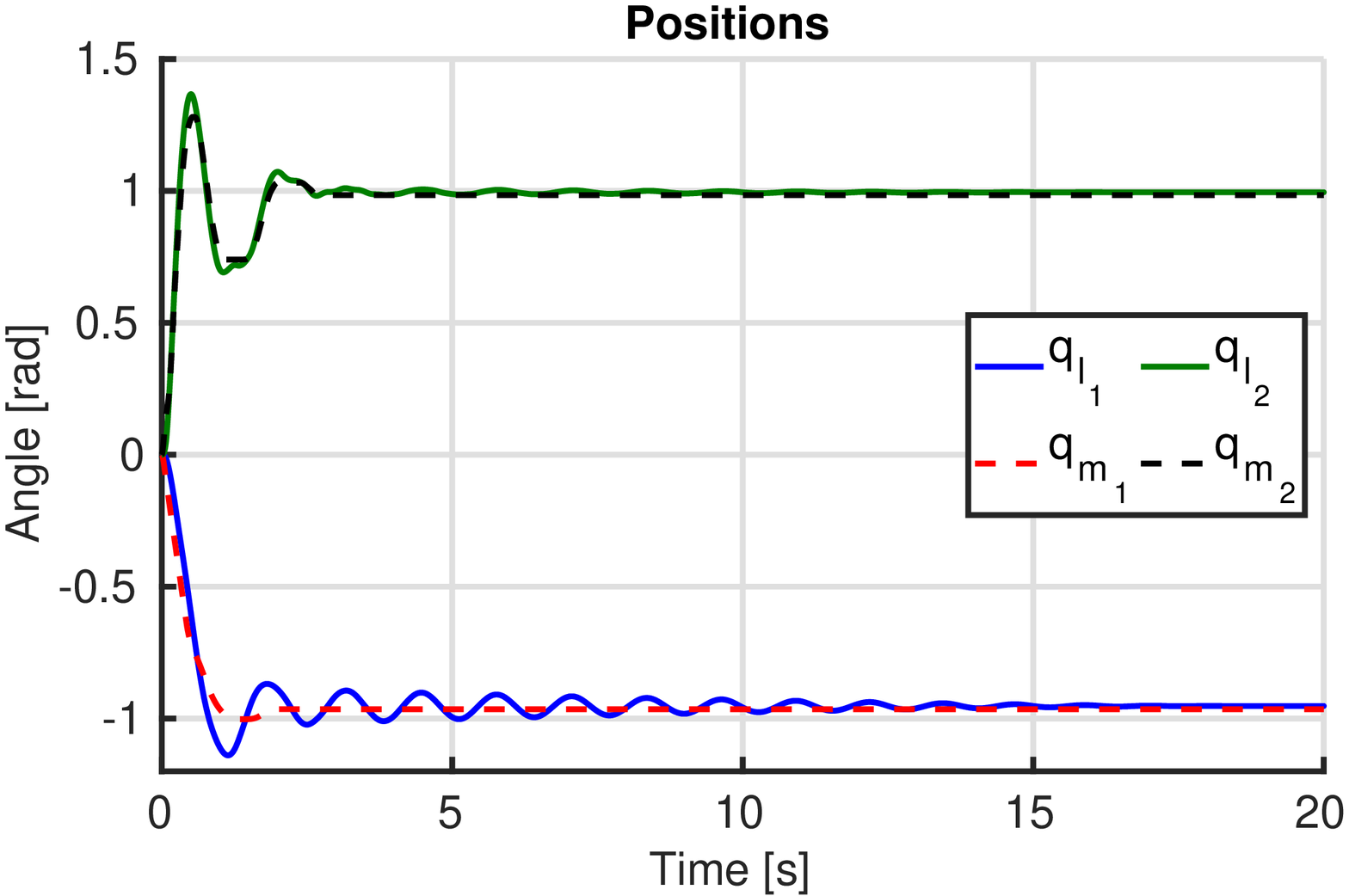}
 \caption{Behavior of the positions for \textbf{C3}.}
 \label{fig:az}
\end{figure}
\begin{figure}
 \centering
 \includegraphics[width=.4\textwidth]{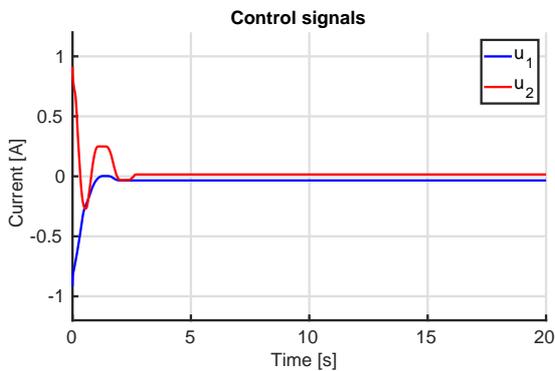}
 \caption{Control signals for case \textbf{C3}.}
 \label{fig:u_azero}
\end{figure}
From a theoretical point of view, the control law \eqref{uzeta} solves the problem stated in Section \ref{sec:prem}, that is, the controller developed in this section is saturated and does not depend on the measurement 
of velocities. Nonetheless, as is reported above, a new issue has arisen during the implementation of the controller, namely, the presence of steady state error. In the following subsection we propose an alternative to 
overcome this problem. 

\subsection{Eliminating the steady state error}
Customarily, the steady state error is eliminated by the negative feedback of an integral term of the error between some measurements and the reference. In our case, such an error is the difference between the positions of the motors and the reference, that is, $(q_{m}-q_{*})$. Moreover, the derivative of this error is given by $\dot{q}_{m}$, which is the passive output of the system. Hence, following the approach adopted in this note, the term that eliminates the steady state error is given by a double integrator of the passive output. Accordingly, Proposition \ref{prop:int} 
provides an alternative control law that includes a double integral-like term, which eradicates the steady state error and ensures that the control signals remain saturated, alas, the pH structure is not preserved anymore.

\begin{proposition}\label{prop:int}
Consider the vector state $\sigma\in \rea_{2}$ whose dynamics are given by
\begin{equation}
 \dot{\sigma}=\nabla^{2} \Phi_{\sigma}(\sigma)(q_{m}-q_{*})-K_{\sigma}\sigma,
\end{equation}
where the matrix $K_{\sigma}\in \rea^{2\times 2}$ and the function $\Phi_{\sigma}:\rea^{2}\to \rea$ are defined as
\begin{eqnarray}
 \Phi_{\sigma}(\sigma)&:=&\displaystyle\sum_{i=1}^{2}\frac{\alpha_{\sigma_{i}}}{\beta_{\sigma_{i}}}\ln(\cosh(\beta_{\sigma_{i}}\sigma_{i}))
 \label{Phis}\\
 K_{\sigma}&:=& \diag\{k_{\sigma_{1}}, k_{\sigma_{2}} \} \label{ksig}
\end{eqnarray} 
with $\alpha_{\sigma_{i}},\beta_{\sigma_{i}}, k_{\sigma_{i}}\in \rea_{>0}$.
Consider $F_{\zeta}$, given in \eqref{Fzeta} and assume that $R_{c_{l}}$ verifies \eqref{schurzeta}. Consider the control law 
\begin{equation}
 \begin{array}{rcl}
  u&=&-\nabla_{z_{c_{l}}}\Phi_{l}(z_{l}(q_{l},x_{c_{l}}))-\nabla_{z_{c_{m}}}\Phi_{m}(z_{m}(q_{m},x_{c_{m}}))\\ &&-\nabla \Phi_{\sigma}(\sigma),
 \end{array}\label{uxi}
\end{equation}
with $z_{l},\ z_{m},\ \Phi_{l},\ \Phi_{m}$ defined as in \eqref{zPhi}. Fix a reference $q_{*}$ and define the matrices
\begin{equation}
 \begin{array}{l}
  A_{\sigma}:=\diag\{\beta_{\sigma_{1}}\alpha_{\sigma_{1}}, \beta_{\sigma_{2}}\alpha_{\sigma_{2}}\}; \\[.2cm]
  \begin{array}{rl}
   A_{\xi_{1}}:=\begin{bmatrix}
\mathbf{0}_{6} \\ -A_{\sigma} \\                                                                                                                   
              \mathbf{0}_{4}                                                                                          \end{bmatrix}, & A_{\xi_{2}}:=\begin{bmatrix}
\mathbf{0}_{2} \\ A_{\sigma} \\                                                                                                                   
              \mathbf{0}_{8}                                                                                          \end{bmatrix}, 
  \end{array}
  \\[.6cm]
  \mathcal{A}:=\begin{bmatrix}
             F_{\zeta}(\nabla^{2} H_{\zeta})_{*} & A_{\xi_{1}} \\ A_{\xi_{2}}^{\top} & -K_{\sigma}   
              \end{bmatrix},
 \end{array}\label{A}
\end{equation}
where $(\nabla^{2} H_{\zeta})_{*}$ is given in \eqref{hesszeta}.
Then:
\begin{itemize}
 \item[(i)] The elements of the input vector $u$ are saturated.
 \item[(ii)] Consider system \eqref{sys} in closed-loop with \eqref{uxi}. If the matrix $\mathcal{A}$ is Hurwitz, then $\xi_{*}=[q_{*}^{\top} \ q_{*}^{\top} \ \mathbf{0}_{10}^{\top}]^{\top}$ 
is a (\textit{locally}) asymptotically stable equilibrium point of the closed-loop system.
\end{itemize}
\end{proposition}
\begin{proof}
To proof (i), note that
\begin{equation}
 \nabla\Phi_{\sigma}=\begin{bmatrix}
                  \alpha_{\sigma_{1}}\tanh(\beta_{\sigma_{1}}\sigma_{1}) \\ \alpha_{\sigma_{2}}\tanh(\beta_{\sigma_{2}}\sigma_{2}) \label{nPhis}
                 \end{bmatrix}.
\end{equation} 
Furthermore, from the proof of item (i) in Proposition \ref{prop:sat} and \eqref{nPhis}, we get
 \begin{equation*}
\begin{array}{rcl}
  u_{i}&=&-\alpha_{l_{i}}\tanh(\beta_{l_{i}}z_{l_{i}})-\alpha_{m_{i}}\tanh(\beta_{m_{i}}z_{m_{i}})\\&&-\alpha_{\sigma_{i}}\tanh(\beta_{\sigma_{i}}\sigma_{i})
\end{array}
\end{equation*} 
Moreover, 
\begin{equation*}
 -(\alpha_{l_{i}}+\alpha_{m_{i}}+\alpha_{\sigma_{i}})\leq u_{i} \leq \alpha_{l_{i}}+\alpha_{m_{i}}+\alpha_{\sigma_{i}}.
\end{equation*} 
To proof (ii), define the new state space $\xi:=[\zeta^{\top} \ \sigma^{\top}]^{\top}$, and the error $\bar{\xi}:=\xi-\xi_{*}$. Then, some lengthy but straightforward computations show that the linearization of closed-loop system, around $\xi_{*}$, is given by
\begin{equation*}
 \dot{\bar{\xi}}=\mathcal{A}\bar{\xi}.
\end{equation*} 
The proof is completed by applying Lyapunov's Indirect Method, see Chapter 4 of \cite{KHA}. 
\end{proof}
While the main theoretical contributions of this document are the results of Proposition \ref{prop:sat}, with the implementation of the control law \eqref{uxi} the closed-loop system exhibits a better performance in terms of steady state error and oscillations. 
Below, we report the experimental results of this implementation.
\subsubsection*{\bf{Experimental results}}
the control law \eqref{uxi} is implemented in the robot arm. Towards this end, we select the control matrices
as
\begin{equation}
 \begin{array}{rclrcl}
R_{c_{l}}&=&\diag\{ 25, 40\}, &  K_{c}&=&0.1I_{2}, \\[.1cm] R_{c_{m}}&=&0.25I_{2}, & K_{\sigma}&=&I_{2},\end{array}
\end{equation} 
The rest of the control parameters is given in Table \ref{cpar}. 
\begin{table}[!h]
\caption{Control parameters}
\label{cpar}
\vspace{.2cm}
\begin{center}
\begin{tabular}{|c|c|c|}
\hline
$\alpha_{l_{1}}=0.6$ & $\alpha_{l_{2}}=0.3$ & $\alpha_{m_{1}}=0.25$  \\\hline
$\alpha_{m_{2}}=0.6$ & $\alpha_{\sigma_{1}}=0.35$ & $\alpha_{\sigma_{2}}=0.3$ \\\hline
$\beta_{l_{1}}=3$ & $\beta_{l_{2}}=1$ & $\beta_{m_{1}}=1$ \\\hline
$\beta_{m_{2}}=1$ & $\beta_{\sigma_{1}}=2.5$ & $\beta_{\sigma_{2}}=3$ \\\hline
\end{tabular}
\end{center}
\end{table}

Moreover, we fix the reference $q_{*}=(-1,1)$, and we carry out the experiments under initial condition $\xi(0)=\mathbf{0}_{14}$.
Note that with this selection of the parameters of the controller, we ensure that the matrix $\mathcal{A}$, defined in \eqref{A}, is Hurwitz and consequently the point $\zeta_{*}$ is an asymptotically stable equilibrium of the closed-loop system. 
Figure \ref{fig:p_int} depicts the positions of the links and motors, where the error of the final position with respect to the desired reference is almost zero. On the other hand, Figure \ref{fig:u_int} shows the control signals, where the value of both signals is close to zero. From the plots we conclude the implementation of controller \eqref{uxi} has in overall a better performance than controller \eqref{uzeta}. Moreover, in this case the trade off between the attenuation of the oscillations and the steady error seems to have been removed. The errors that persist can be result of several factors, e.g., slips in the motors, numerical errors, the tuning of the gains. 

\begin{figure}
 \centering
 \includegraphics[width=.45\textwidth]{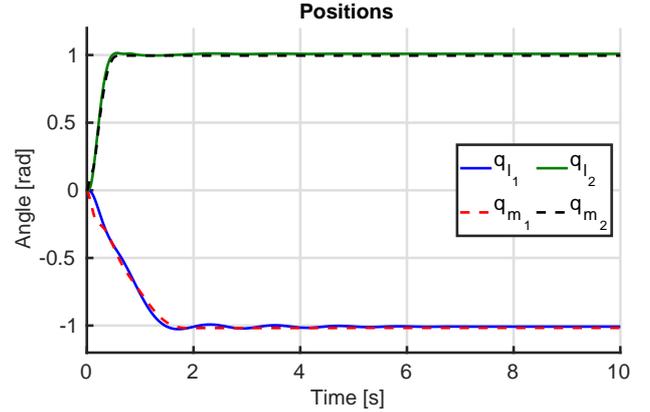}
 \caption{Behavior of the positions with the addition of an integral-like term.}
 \label{fig:p_int}
\end{figure}
\begin{figure}[!h]
 \centering
 \includegraphics[width=.45\textwidth]{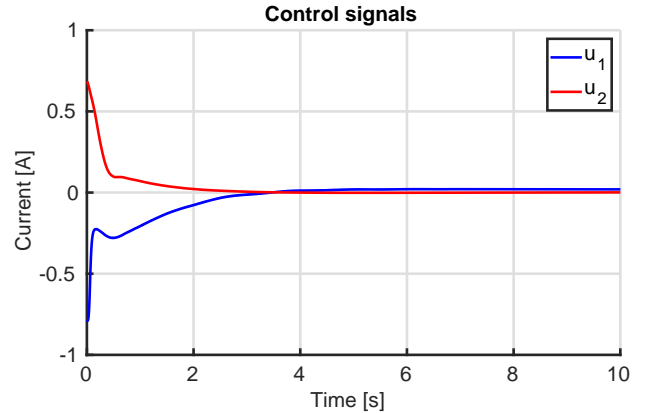}
 \caption{Control signals with the addition of an integral-like term.}
 \label{fig:u_int}
\end{figure}

\section{Conclusions and future work}
\label{sec:cfw} 
We present a controller that solves the stabilization problem for planar robots with flexible joints. Additionally, the control law is designed without the necessity of solving PDEs nor velocity measurements; and the control signals are saturated. Furthermore,
the controller is designed for the two links case, nonetheless, can be extended to the case of $n$ links. 

We remark the fact that by considering the natural damping of the system we are able to inject damping in the links and preserve the pH structure.

For implementation purposes, we add an extra term dependent on a virtual state which is similar to an integral of the error $(q_{m}-q_{*})$. This extra term improve the performance of the closed-loop system with the elimination of the steady state error and the attenuation of the oscillations in its behavior. Furthermore, after the addition of this extra term, the property of saturation in the control law is preserved.

As future work, we aim to extend the proposed methodology to pH systems, in different domains, that can be stabilized with PI-PBC, e.g., electrical circuits, electro-mechanical systems or fluid systems.
\appendix
Table \ref{parameters} contains the information about the physical parameters of the robot arm \textit{2 DOF serial flexible joint} by Quanser. These parameters were taken from the datasheet of the robot and \cite{MIR}.
\begin{table}[!h]
\caption{System parameters}
\label{parameters}
\small
\begin{center}
\begin{tabular}{|c|c|c|c|}
\hline
Parameter & Physical meaning & Value & Units\\\hline
$a_{1}$ & × & $0.148$ & $\left[kg\cdot m^{2} \right]$\\\hline
$a_{2}$ & × & $0.073$ & $\left[kg\cdot m^{2} \right]$\\\hline
$b$ & × & $0.086$ & \small{$\left[kg\cdot m^{2} \right]$}\\\hline
$\mathcal{I}_{m_{1}}$ & MoI of motor $1$& $0.217$ & $\left[kg\cdot m^{2} \right]$\\\hline
$\mathcal{I}_{m_{2}}$ & MoI of motor $2$& $0.007$ & $\left[kg\cdot m^{2} \right]$\\\hline
$D_{l_{1}}$ & Damping on link $1$& $0.038$ & $\left[N\cdot m\cdot s / rad\right]$ \\\hline
$D_{l_{2}}$ & Damping on link $2$& $0.03$ & $\left[N\cdot m\cdot s / rad\right]$ \\\hline
$D_{m_{1}}$ & Damping on motor 1& $8.435$ & $\left[N\cdot m\cdot s / rad\right]$\\\hline
$D_{m_{2}}$ & Damping on motor 2& $0.136$ & $\left[N\cdot m\cdot s / rad\right]$\\\hline
$k_{s_{1}}$ & Spring constant 1 & $9$ & $\left[N \cdot m/rad\right]$\\\hline
$k_{s_{2}}$ & Spring constant 2 & $4$ & $\left[N \cdot m/rad\right]$ \\\hline
\end{tabular}
\end{center}
\normalsize
\end{table}

The experiments reported in Section \ref{sec:cd} were carried out using the Matlab/Simulink interface. Where, 
\begin{equation*}
 u_{\tt max_{1}}=u_{\tt max_{2}}=1.2[A].
\end{equation*} 
Note that, the saturation value is given in terms of the currents supplied to the motors, nonetheless, these currents satisfy a linear relationship with the torques of the motors. The details of these linear relations are provided by Quanser.  
\bibliographystyle{IEEEtran}
\bibliography{bibliographyifac}

\begin{thebibliography}{10}
\providecommand{\url}[1]{#1}
\csname url@samestyle\endcsname
\providecommand{\newblock}{\relax}
\providecommand{\bibinfo}[2]{#2}
\providecommand{\BIBentrySTDinterwordspacing}{\spaceskip=0pt\relax}
\providecommand{\BIBentryALTinterwordstretchfactor}{4}
\providecommand{\BIBentryALTinterwordspacing}{\spaceskip=\fontdimen2\font plus
\BIBentryALTinterwordstretchfactor\fontdimen3\font minus
  \fontdimen4\font\relax}
\providecommand{\BIBforeignlanguage}[2]{{%
\expandafter\ifx\csname l@#1\endcsname\relax
\typeout{** WARNING: IEEEtran.bst: No hyphenation pattern has been}%
\typeout{** loaded for the language `#1'. Using the pattern for}%
\typeout{** the default language instead.}%
\else
\language=\csname l@#1\endcsname
\fi
#2}}
\providecommand{\BIBdecl}{\relax}
\BIBdecl

\bibitem{BULLO}
F.~Bullo and A.~D. Lewis, \emph{Geometric control of mechanical systems:
  modeling, analysis, and design for simple mechanical control systems}.\hskip
  1em plus 0.5em minus 0.4em\relax Springer Science \& Business Media, 2004,
  vol.~49.

\bibitem{spong1987}
M.~W. Spong, ``Modeling and control of elastic joint robots,'' \emph{Journal of
  dynamic systems, measurement, and control}, vol. 109, no.~4, pp. 310--318,
  1987.

\bibitem{olfati2001}
R.~Olfati-Saber, ``Nonlinear control of underactuated mechanical systems with
  application to robotics and aerospace vehicles,'' Ph.D. dissertation,
  Massachusetts Institute of Technology, 2001.

\bibitem{ORTbook}
R.~Ortega, A.~Loria, P.~Nicklasson, and H.~Sira-Ramirez,
  \emph{{P}assivity-{B}ased {C}ontrol of {E}uler-{L}agrange {S}ystems:
  {M}echanical, {E}lectrical and {E}lectromechanical {A}pplications}.\hskip 1em
  plus 0.5em minus 0.4em\relax Communications and Control Engineering. Springer
  Verlag, London, 1998.

\bibitem{VAN}
A.~J. van~der Schaft, \emph{$L_2$-{G}ain and {P}assivity techniques in
  nonlinear control}, 3rd~ed.\hskip 1em plus 0.5em minus 0.4em\relax Berlin:
  Springer, 2016.

\bibitem{ORTtac}
R.~Ortega, M.~W. Spong, F.~Gomez-Estern, and G.~Blankenstein, ``Stabilization
  of a {C}lass of {U}nderactuated {M}echanical {S}ystems {V}ia
  {I}nterconnection and {D}amping {A}ssignment,'' \emph{Automatic Control, IEEE
  Transactions on}, vol.~47, no.~8, pp. 1218--1233, Aug 2002.

\bibitem{jardon}
H.~Jard{\'o}n-Kojakhmetov, M.~Mu{\~n}oz-Arias, and J.~M. Scherpen, ``Model
  reduction of a flexible-joint robot: a port-hamiltonian approach,''
  \emph{IFAC-PapersOnLine}, vol.~49, no.~18, pp. 832--837, 2016.

\bibitem{BORCISORT}
P.~Borja, R.~Cisneros, and R.~Ortega, ``A constructive procedure for
  energy-shaping of port-{H}amiltonian systems,'' \emph{Automatica}, vol.~72,
  pp. 230--234, 2016.

\bibitem{KHA}
H.~Khalil, \emph{Nonlinear systems}, 3rd~ed.\hskip 1em plus 0.5em minus
  0.4em\relax New Jersey: Prentice-Hall, 2002.

\bibitem{dirksz}
D.~A. Dirksz, J.~M. Scherpen, and R.~Ortega, ``Interconnection and damping
  assignment passivity-based control for port-{H}amiltonian mechanical systems
  with only position measurements,'' in \emph{Decision and Control, 2008. CDC
  2008. 47th IEEE Conference on}.\hskip 1em plus 0.5em minus 0.4em\relax IEEE,
  2008, pp. 4957--4962.

\bibitem{MIR}
R.~Miranda-Colorado, L.~T. Aguilar, and J.~Moreno-Valenzuela, ``A model-based
  velocity controller for chaotization of flexible joint robot manipulators:
  Synthesis, analysis, and experimental evaluations,'' \emph{International
  Journal of Advanced Robotic Systems}, vol.~15, no.~5, p. 1729881418802528,
  2018.

\end{thebibliography}

\end{document}